\documentclass[11pt]{article}
\usepackage{graphicx}
\usepackage{latexsym}
\usepackage{amssymb}
\usepackage{amsmath}
\usepackage{amsthm}
\usepackage{bm}

\usepackage{mathtools}

\usepackage{nth}

\usepackage{forloop}
\usepackage[paper=letterpaper,margin=1in]{geometry}
\usepackage{times}
\usepackage{paralist}
\usepackage{nicefrac}
\usepackage{rotating}
\usepackage{array,multirow}
\usepackage{threeparttable}
\usepackage{hyperref}
\usepackage{etoolbox}
\usepackage{tikz,tkz-graph}
\usetikzlibrary{arrows,shapes}
\usetikzlibrary{positioning,calc}
\usetikzlibrary{decorations.pathreplacing}
\usetikzlibrary{fadings}
\tikzset{>= stealth'}
\tikzstyle{vertex}=[circle, draw,fill=gray!20, inner sep=0pt, minimum size=16pt]
\tikzstyle{square}=[rectangle, draw,fill=gray!20, inner sep=0pt, minimum size=16pt]
\tikzstyle{svertex}=[circle, draw,fill=gray!20, inner sep=0pt, minimum size=12pt]
\tikzstyle{ssquare}=[rectangle, draw,fill=gray!20, inner sep=0pt, minimum size=12pt]
\tikzstyle{tvertex}=[circle, draw,fill=gray!20, inner sep=0pt, minimum size=8pt]
\tikzstyle{tsquare}=[rectangle, draw,fill=gray!20, inner sep=0pt, minimum size=6pt]

\usepackage{environ}
\makeatletter
\NewEnviron{Lalign}{\tagsleft@true\begin{align}\BODY\end{align}}
\NewEnviron{Lalign*}{\tagsleft@true\begin{align*}\BODY\end{align*}}
\makeatother

\newbool{shortVersion}

\newtheorem{theorem}{Theorem}[section]
\newtheorem{lemma}[theorem]{Lemma}

\newtheorem{corollary}[theorem]{Corollary}

\newtheorem{observation}[theorem]{Observation}
\newtheorem{claim}[theorem]{Claim}

\makeatletter
\newtheorem*{rep@theorem}{\rep@title}
\newcommand{\newreptheorem}[2]{%
\newenvironment{rep#1}[1]{%
 \def\rep@title{#2 \ref{##1}}%
 \begin{rep@theorem}}%
 {\end{rep@theorem}}}
\makeatother

\newreptheorem{theorem}{Theorem}

\newcommand\E{\mathbb{E}}

\newcommand\pE{\mathbb{\tilde E}}
\newcommand{\ov}{\overline}

\newcommand{\defcal}[1]{\expandafter\newcommand\csname c#1\endcsname{{\mathcal{#1}}}}
\newcommand{\defbb}[1]{\expandafter\newcommand\csname b#1\endcsname{{\mathbb{#1}}}}
\newcounter{calBbCounter}
\forLoop{1}{26}{calBbCounter}{
    \edef\letter{\Alph{calBbCounter}}
    \expandafter\defcal\letter
		\expandafter\defbb\letter
}

\newcommand{\nb}[1]{{\color{blue}NB: #1}}

\newcommand{\oalg}{{\textrm{\small ALG}}}
\newcommand{\orel}{{\textrm{\small REL}}}
\newcommand{\lbE}{{\mathrm{LB}(\emptyset)}}
\newcommand{\lbJ}{{\mathrm{LB}(J)}}
\newcommand{\lbG}{{\mathrm{LB}({G})}}

\allowdisplaybreaks

\title{Lift-and-Round to Improve Weighted Completion Time on Unrelated Machines}

\author{Nikhil Bansal%
\thanks{Department of Mathematics and Computer Science, TU Eindhoven, Netherlands.  
Email:
\href{mailto:n.bansal@tue.nl}{n.bansal@tue.nl}. 
Supported by NWO Vidi grant 639.022.211 and ERC consolidator grant 617951.}
\and
Aravind Srinivasan%
\thanks{Department of Computer Science and Institute for Advanced Computer Studies, University of Maryland, USA. 
Email:
\href{mailto:srin@cs.umd.edu}{srin@cs.umd.edu}.
Supported in part by NSF Awards CNS-1010789 and CCF-1422569, and a research award from Adobe, Inc.}
\and
Ola Svensson\thanks{School of Computer and Communication Sciences, EPFL, Switzerland. 
Email:
\href{mailto:ola.svensson@epfl.ch}{ola.svensson@epfl.ch}.
Supported by ERC Starting Grant 335288-OptApprox.}
}

\date{\today}

\begin{document}

\maketitle


\begin{abstract}

We consider the problem of scheduling  jobs on unrelated machines so as to
minimize the sum of weighted completion times. Our main result is a $(3/2-c)$-approximation algorithm for some fixed $c>0$, improving upon the long-standing bound of 3/2 (independently due to Skutella, \emph{Journal of the ACM}, 2001, and Sethuraman \& Squillante, \emph{SODA}, 1999). 
To do this, we first introduce a new lift-and-project based SDP relaxation for the problem. 
This is necessary as the previous convex programming relaxations  
have an integrality gap of $3/2$. 
Second, we give a new general bipartite-rounding procedure 
that produces an assignment with certain strong negative correlation properties.
\end{abstract}


%
%

\pagenumbering{arabic}


\section{Introduction}
We consider the classic problem of scheduling  jobs on unrelated machines to
minimize the sum of weighted completion times.  Formally, a problem instance
consists of a set $J = \{1, 2, \dots, n\}$ of $n$ jobs and a set $M$ of $m$
machines; each job $j\in J$ has a weight $w_j \geq 0$ and it requires
a processing time of $p_{ij}\geq 0$ if assigned to machine $i\in M$. The goal
is to find a schedule that minimizes the weighted completion time, that is
$\sum_{j\in J} w_j C_j$, where $C_j$ denotes the completion time of job $j$ in
the schedule constructed.

Total completion time and related metrics such as makespan and flow time, are
some  of the most relevant and well-studied measures of quality of service in
scheduling and resource allocation problems.  While total completion time has
been studied since the 50's \cite{Smith56}, a systematic study of its
approximability was started in the late 90's by \cite{PSW98}. This led to
a lot of activity and progress on the problem in various scheduling models and
settings (such as with or without release dates, preemptions, precedences,
online arrivals etc.). In particular, we have now  a complete understanding of
the approximability  in \emph{simpler} machine models, such as identical and
related machines.
For these settings, 
non-trivial approximation schemes were developed more than a decade ago, e.g., in
\cite{Afrati99, SkutellaW99, ChekuriK01}.  The more general unrelated machine
model behaves very differently and is significantly more challenging.
Perhaps because of this,  its study  has led to the development of many new
techniques, such as interesting LP and convex programming formulations and
rounding techniques~\cite{PSW98,ChekuriMNS01, HallSSW97, SchulzS02, Goemans02,
Skutella01} (see also the survey by Chekuri and Khanna \cite{ChekuriK04}).

In spite of these impressive developments, it remains a notorious problem to
understand the  approximability of total weighted completion time in the
unrelated machines setting. On the
positive side, Schulz and Skutella \cite{SchulzS02} gave a $(3/2+\epsilon)$-approximation based on a time-indexed LP formulation, improving upon the
previous works of \cite{PhillipsSW97,HallSW96}. Another $3/2$-approximation was
obtained by Skutella \cite{Skutella01} and independently by Sethuraman and
Squillante \cite{SethuramanS99} based on a novel convex-programming relaxation.
On the other hand, Hoogeveen et al.~\cite{HoogeveenSW98} showed that the
problem is APX-hard, but the hardness factor was very close to $1$.  The
natural question of whether a $(3/2-c)$-approximation exists for the problem for
some $c>0$ has been proposed widely \cite{ChekuriK04, SchulzS02, KumarMPS08,
SviridenkoW13, Schuurman99}. In particular, it appears as {\em Open Problem 8}
in the well-known list \cite{Schuurman99} due to Schuurman and Woeginger of the
``top ten" problems in scheduling. Moreover, Srividenko and Wiese  conjectured
\cite{SviridenkoW13} that the configuration LP for this problem  has an integrality
gap that is strictly less than $3/2$.  The unrelated machines setting is one of
the most general and versatile scheduling models that incorporates the heterogeneity of
jobs and machines. But besides the practical motivation, an important reason
for interest in the problem is that historically, the exploration of various
scheduling problems in the unrelated machines model has been a rich source of
several new algorithmic techniques \cite{LenstraST90, Svensson12, BansalS06,
ChakrabartyCK09, AsadpourS07, AsadpourFS12, Feige08, MakarychevS14, KumarMPS09,
AzarE05}.

\paragraph{Our Results:}
Our main result is such an improved algorithm. In particular, we show the following.
\begin{theorem}
\label{thm:intro-main}
  There is a  $(\nicefrac{3}{2} - c)$-approximation algorithm for minimizing the total weighted completion time on unrelated machines, for some $c \geq 10^{-7}$. 
\end{theorem}

Remark: We do not try to optimize the constant $c$ too much, preferring instead to keep the exposition as simple as possible. However, it does not seem likely that our analysis would yield $c < 10^{-2}$.	

The result is based on two key ideas: (i) a novel SDP relaxation for the problem obtained by applying one round of lift-and-project to the standard LP	
formulation and (ii) a new rounding algorithm to assign jobs to machines that reduces the ``correlation" between the various jobs assigned to a machine.

\vspace{2mm}
\noindent{\bf Stronger Formulation:} 
The stronger formulation is necessary: we show that the convex programming relaxation considered by \cite{Skutella01,SethuramanS99} has an integrality gap of $3/2$. Such families of instances do not seem to have been previously known \cite{Skutella:personal}; we describe them in Section \ref{s:prel}. 
In contrast, our stronger relaxation can be used to derive new and tighter lower bounds on the optimum value, by exploiting the PSD constraint on the underlying moment matrix.

We remark that the lower bounds we use to prove Theorem \ref{thm:intro-main} can also be obtained using the configuration LP proposed by \cite{SviridenkoW13}, which confirms their conjecture that the integrality gap of the configuration LP is also upper bounded by $(3/2-c)$. However, we find the SDP formulation more natural as it is explicitly reveals the correlation information that we use.

\vspace{2mm}
\noindent{\bf The Rounding Algorithm:} 
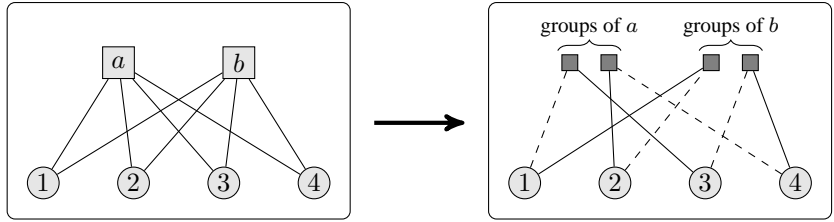
\begin{figure}[t!]
  \centering
    \begin{tikzpicture}[scale=0.8]
    \node[ssquare] (m1) at (0.75,0) {\small $a$};
    \node[ssquare] (m2) at (2.75,0) {\small $b$};
    \foreach \val in {1,2,3,4} {
      \node[svertex] (j\val) at (-2 + 1.5*\val, -2) {\small $\val$};
      \path (j\val) edge (m1) edge (m2);
    }
    \draw[rounded corners=3pt] (-1.1,1) rectangle (4.6,-2.6);
    \draw (5, -1) edge[->,ultra thick] (6.5,-1);
    \begin{scope}[xshift = 8cm]
      \draw [decorate, decoration={brace, amplitude=4pt}] (0.05, 0.2) -- node[above =0.05cm] { \scriptsize groups of $a$} (1.1, 0.2);
      \draw [decorate, decoration={brace, amplitude=4pt}] (2.4, 0.2) -- node[above =0.05cm] { \scriptsize groups of $b$} (3.45, 0.2);
      \node[tsquare, fill=gray] (m11) at (0.25,0) {};
      \node[tsquare, fill=gray] (m12) at (0.9,0) {};
      \node[tsquare, fill=gray] (m21) at (2.6,0) {};
      \node[tsquare, fill=gray] (m22) at (3.25,0) {};
      \foreach \val in {1,2,3,4} {
        \node[svertex] (j\val) at (-2 + 1.5*\val, -2) {\small $\val$};
      }
      \path (j1) edge[dashed] (m11) edge (m21);
      \path (j2) edge[] (m12) edge[dashed] (m21);
      \path (j3) edge[] (m11) edge[dashed] (m22);
      \path (j4) edge[dashed] (m12) edge[] (m22);
      \draw[rounded corners=3pt] (-1.1,1) rectangle (4.6,-2.6);
    \end{scope}
  \end{tikzpicture}
  \caption{A simple example motivating the novel rounding algorithm with strong negative correlation.}
  \label{fig:roundingEX}
\end{figure}
The solution to the SDP gives a fractional assignment of jobs to machines,
which we need to convert to an integral assignment. Interestingly, all the
previous algorithms \cite{SchulzS02,Skutella01,SethuramanS99} are based on
applying standard (i.e., independent across jobs) randomized rounding to the
fractional solution to find an assignment of jobs to machines.
However, a very simple example (in Section \ref{s:prel}) shows that no such
``\emph{independent} randomized rounding" based algorithm can give a $3/2
- \Omega(1)$ guarantee, irrespective of the underlying convex relaxation.
The problem is that the variance can be too high.

To get around this, we need to introduce some {\em strong} negative correlation
among pairs of jobs assigned to any machine $i$ (i.e., the ratio of the probability that they are both scheduled on $i$ to the product of their respective
probabilities of assigment on $i$, should be $1 - \Omega(1)$). For intuition, consider the example
depicted on the left in Figure~\ref{fig:roundingEX}: we have a set $\{1,2,3,4\}$
of four jobs, two machines $a,b,$ and the SDP assigns each job fractionally
$\nicefrac{1}{2}$ to both machines. Note that independent randomized rounding
would assign any two jobs $j$ and $j'$ to machine $a$ (and similarly to $b$)
with probability $\nicefrac{1}{4}$. Ideally, we
would like to have strong negative correlation that decreases this probability for all pairs of jobs. Unfortunately, this is not possible in general as
can be seen by taking $n\gg1$ jobs instead of $4$ in the considered
example\footnote{Any schedule $\pi$ of $n$ jobs on $2$ machines, must
  have $\Pr_{j, j'}[\mbox{ $\pi$ assigns $j$ and $j'$ to the same machine}]
  \geq 1/2 - o(1)$. A simple proof of this is as follows. Suppose $\pi$ assigns
  $s$ jobs to the first machine and $t$ jobs to the second machine, where
  $s+t=n$. Then, conditional on this, the desired probability is $({s \choose
  2} + {t \choose 2})/{n \choose 2}$ which is minimized at $s=t=n/2$, and has
  value $1/2 -o(1)$. Now as $\E_\pi \Pr_{j, j'}[\mbox{ $\pi$ assigns $j$ and $j'$ to the same machine}] = \E_{j,j'} \Pr_{\pi}[\mbox{ $\pi$ assigns $j$ and $j'$ to the same machine}]$, we have that there exist two jobs $j$ and $j'$ that are assigned to the same machine with probability at least $1/2-o(1)$ (and thus to one of them with probability  at least $1/4 - o(1)$) no matter which algorithm, i.e., distribution over schedules $\pi$, that is used. }. 
However,  one can still hope for a randomized rounding with strong negative correlation
for some of the jobs while maintaining  that no two jobs are assigned to a single
machine with probability more than $\nicefrac{1}{4}$. 
This is what our randomized rounding algorithm achieves. 

The pairs of jobs that will have strong negative correlation are decided by
a grouping scheme: for each machine $i$, the jobs are partitioned  into groups 
with total fractional assignment on $i$ being at most $1$. The jobs in the same group
are those that will have strong negative correlation.  This step is illustrated
on the right of Figure~\ref{fig:roundingEX}. Machine $a$ has two groups
consisting of jobs $\{1,3\}$ and $\{2,4\}$ and  machine $b$ has two groups
consisting of jobs $\{1,2\}$ and $\{3,4\}$.
Viewing this as a bipartite graph with group and job vertices, we would like
to find an assignment with strong negative correlation on the edges incident to
the same group.  This is reminiscent of the several randomized pipage-based
schemes \cite{AgeevS99,AroraFK02,ChekuriVZ11,GandhiKPS06,KahnK97} that given
a fractional matching produce an integral matching. In fact, these get perfect
negative correlation between edges at a vertex as only one edge is
picked at any vertex\footnote{Although they can introduce positive correlation between
non-adjacent edges.}.
 
However, these techniques do not work in our setting of general assignments due
to a somewhat subtle issue; trying to force strong negative correlation between
two edges in a group of a machine can cause unexpected positive correlations
among other edges of that machine. In particular, in our example, previous
rounding techniques would output one of the two perfect matchings with equal
probability (the two perfect matchings are indicated by dashed and solid edges
in the right of Figure~\ref{fig:roundingEX}) -- thus yielding perfect
\emph{positive} correlation, e.g., for jobs $1$ and $4$ being assigned to machine
$a$.


To get around this we give a new rounding theorem. 
The main idea behind the algorithm is to update the fractional assignment using
randomized pipage steps along carefully chosen paths of length 4. In
particular, these paths are chosen based on a random 2-coloring of the edges
where the coloring is based on the fractional assignment and evolves over time.
The properties of this general rounding technique are summarized in the
theorem below. We believe that
this technique can be of independent interest, as it appears to be the first to
obtain \emph{strong} negative correlations. Indeed, our rounding maintains the
desired properties from independent randomized rounding (properties (\ref{genprop1}), (\ref{genprop2}), and the
 second part of (\ref{genprop3}) of Theorem~\ref{thm:genassignment}) while also achieving
a guaranteed amount of pairwise negative -- i.e., strong -- correlation (the first part of (\ref{genprop3})). This is key for our
result, and we are not aware of any prior work in this vein.

\begin{theorem}
  Let $\zeta = \nicefrac{1}{108}$.
  Consider a bipartite graph $G=(U \cup V, E)$ and let $y \in [0,1]^E$ be
  fractional values on the edges satisfying $ y(\delta(v)) = 1$ for all $v\in
  V$. For each  vertex $u\in U$, select any family of disjoint $E_u^{(1)},
  E_u^{(2)}, \ldots, E_u^{(\kappa_u)} \subseteq \delta(u)$ subsets of edges
  incident to $u$ such that $y(E_u^{(\ell)}) \leq 1$ for $\ell =1 ,\ldots,
  \kappa_u$. Then there exists a   randomized  polynomial-time algorithm
  that outputs a random subset of the edges $E^*\subseteq E$ satisfying
  \begin{enumerate}[(a)]
    \item \label{genprop1} For every $v\in V$, we have $|E^* \cap \delta(v)| = 1$ with probability $1$;
    \item \label{genprop2} For every $e \in E$, $\Pr[e\in E^*] = y_e$;
    \item \label{genprop3} For every $u\in U$ and all $e\neq e'\in \delta(u)$,
      \begin{align*}
        \Pr[ e\in E^* \wedge e' \in E^*] \leq
        \begin{cases}
          (1-\zeta) \cdot y_ey_{e'} & \mbox{if } e,e' \in E_u^{(\ell)} \mbox{ for some }  \ell \in \{1,2, \ldots, \kappa_u\},\\
          y_e y_{e'} & \mbox{otherwise}.
        \end{cases}
      \end{align*}
  \end{enumerate}
  \label{thm:genassignment}
\end{theorem}

In the above theorem, we use the standard notation $\delta(w) = \{e\in E: w\in
e\}$ to denote the set of edges incident to a vertex $w$, and let $y(F)
= \sum_{e\in F} y_e$ for any subset $F \subseteq E$ of edges.

\section{Preliminaries and Lower Bounds}
\label{s:prel}

On a single machine, the weighted completion is minimized by ordering the jobs
in non-increasing order of $w_j/p_j$, referred to as the Smith ordering. In the
unrelated machines setting, for each machine $i$ let $\preceq_{i}$ denote the
Smith ordering 
of jobs on machine $i$ (i.e.~$j' \preceq_{i} j$ iff $w_{j'}/p_{ij'} \geq  w_{j}/p_{ij}$).
Given an assignment of jobs to machines,  
the total weighted completion time is simply 
\[   \sum_i  \sum_{j \in J(i)} w_j p_{ij} (\sum_{j' \preceq_i j}  p_{ij'})  \]
where $J(i)$ denotes the set of jobs assigned to machine $i$. 

For each $i\in M$ and
$j\in J$, consider a binary variable $x_{ij}$ that should take value $1$ if and only if job $j$ is assigned to
machine $i$.  Then the  exact quadratic program can be formulated as follows: 
\begin{center}
\begin{minipage}{12cm}
\begin{Lalign*}
\label{qp}
\tag{\mbox{QP}} \quad  \textrm{Minimize} \qquad \sum_{i\in M} \sum_{j\in J} &w_j x_{ij} \left( \sum_{j' \in J: j' \preceq_i j}  p_{ij'}  x_{ij'} \right) \\[2mm]
  \textrm{subject to} \qquad   \sum_{i\in M} x_{ij} &= 1 \qquad \quad ~ ~ ~ \mbox{for all } j \in J, \\
  x & \in\{0,1\}^{M \times N}.
\end{Lalign*}
\end{minipage}
\end{center}
\vspace{2mm}
\noindent{\bf The Convex Programming relaxation of \cite{Skutella01,SethuramanS99}:}
We only describe the relaxation of \cite{Skutella01,SethuramanS99} here 
and refer to \cite{Skutella01} for details on how it is obtained.
They relax the variables $x_{ij}$ in (QP) above to be fractional in $[0,1]$, together with the fact that $x_{ij}^2 = x_{ij}$ for an integral solution and that $c^Tx :=\sum_{i} \sum_j w_j p_{ij} x_{ij}$ is a lower bound on any solution to obtain the following convex relaxation:
\begin{center}
  \begin{minipage}{12cm}
\begin{Lalign*}
  \label{cp}
  \tag{\mbox{CP}} \qquad \textrm{Minimize} \qquad z \\
    \textrm{subject to} \qquad z & \geq \frac{1}{2} c^T x + \frac{1}{2} x^T D x \\
		z  & \geq c^T x  \\  
			\sum_{i\in M} x_{ij} &= 1 \qquad \quad ~ ~ ~ \mbox{for all } j \in J, \\ 
  x & \in[0,1]^{M \times N}. 
	\end{Lalign*}
\end{minipage}
\end{center}
where $x^TDx:= \sum_i (\sum_j w_j (\sum_{j' \prec_i j}  2 p_{ij'} x_{ij'} + p_{ij} x_{ij} ) x_{ij})$ can be shown to be a convex function. 
We will refer to $c^Tx$ and $x^TDx$ as the linear and quadratic terms respectively.

%
%

\vspace{2mm}
\noindent{\bf A 3/2 integrality gap instance for CP:} Consider the following instance. There are $k+1$ jobs, all of weight $1$. The first $k$ jobs are of size (processing time) $1$ each and can only be placed on machine $1$ (i.e.,~have infinite size on other machines). Job $k+1$ has size $k^2$ and can be placed on any machine $2,\ldots,m$, where we let $m = k+1$.


\begin{claim}
The above instance has an integrality gap $3/2 - O(1/k)$ for (CP). 
\end{claim}
\begin{proof}
  First observe that  any integral solution has value greater than $({3}/{2}) k^2$ as the total completion time of the first $k$ jobs is $k(k+1)/2$ while the last job has a completion time of $ k^2 $.

Now, consider the fractional solution where each job $1,\ldots,k$ is assigned to extent $1$ on machine $1$, and job $k+1$ is assigned to each machine $i$ for $i=2,\ldots,m$, to an extent of $1/(m-1) = 1/k$.
We will show that this solution has fractional value at most  $k^2 + k$. 
 
First, the linear term $c^Tx$ is $k+k^2$ ($k$ for the first $k$ jobs and $k^2$ for the big job).
Second, the quadratic term is 
\[x^TDx  = \sum_{j=1}^k (2(j-1)+1)  +  \sum_{i=2}^m \frac{k^2}{(m-1)^2}  =  k^2 + \frac{k^2}{m-1} = k^2 + k. \]
In particular, the first $k$ jobs contribute $k^2$ above, and for the last job each of the $m-1$ machines contributes $  k^2/(m-1)^2$. 
Thus (CP) has objective value at most $k^2+k$.
\end{proof}

Note that in this example, the problem is that both the linear and quadratic bounds are weak on the overall instance. In particular, while the linear bound is exact on the big job, it is very weak on the small jobs. On the other hand, the quadratic term is exact on the small jobs, but very weak on the big job.

\vspace{2mm}
\noindent{\bf Limitation of Independent Randomized Rounding based approaches:} The previous-best approximation algorithms are based on standard (i.e., \emph{independent} across jobs) randomized rounding. We show that no such rounding can beat the approximation guarantee of $3/2$, irrespective of the relaxation.
%
%
Consider the (trivial) instance with $m$ jobs each of which can be placed on
any of the $m$ machines, and with $w_j = p_{ij} = 1$ for all $i, j$. 
The fractional solution $x_{ij}=1/m$ for all $i,j\in [m]$ is a valid solution for any relaxation (as it is can be expressed as a convex combination of $m$ perfect matchings). 
Clearly, the optimal solution assigns one job to each machine and has value $m$. 
However, under independent randomized rounding, for large $m$, the number of jobs assigned to a machine
approaches a Poisson distribution with mean $1$ and so the probability that a machine gets $k$ jobs is $\approx 1/(e \cdot k!)$.
The expected completion time on any machine is thus
\[\approx \sum_{k=0}^\infty \frac{k(k+1)}{2} \cdot \frac{1}{ek!}\]
which is $3/2$ as the first and second moments of Poisson($1$)  are $1$ and $2$ respectively.   
%

\vspace{2mm}
\noindent{\bf The need for negative correlation in different classes:} The above example might suggest that randomized 
rounding performs poorly only when the total mass ($\sum_j x_{ij}$) on a machine $i$ is close to $1$, as intuitively the effect of the variance should be relatively small if there are many jobs. This intuition is indeed true if the jobs are similar to each other in terms of size (processing time) and weight. However, the following example shows that some more care is needed if the jobs are very dissimilar.  
Suppose there are $\ell$ job classes $k=1,\ldots,\ell$, where a class $k$ job has weight $M^k$ and size $M^{-k}$ for some large $M$, and that machine $i$ has $m$ jobs from each class, with $x_{ij}=1/m$ for all jobs $j$. So the total fractional assignment of jobs to $i$ is $\ell$. 
Now, as the Smith ratios are very different, the jobs from different classes have negligible effect on each other: only the individual cost of each class matters, and the fractional cost is $\approx \sum_{k=1}^{\ell} M^{k} M^{-k} = \ell$. Now, if we round each job independently, the expected cost is $3 \ell/2$, and it is not hard to see that to get a $((3/2) -c)$--approximation, we need to get a non-trivial negative correlation in at least $\Omega(c)$ fraction of the classes. 

It turns out that this example is in a sense the worst possible; it
motivates our rounding procedure in Section~\ref{s:alg}.
Roughly speaking, it suffices 
 to partition the jobs in different classes so that the total fractional weight is about $1$, and then try to get some strong negative correlation within jobs of each class.



\section{Strong Convex Relaxation}
\label{sec:relaxation}
In this section, we give a strong convex relaxation based on the paradigm of
``systematically'' relaxing the exact quadratic mathematical program~\eqref{qp}   to
a tractable convex program. In particular, our relaxation can be obtained
``automatically'' using the Lasserre/Sum-of-Squares hierarchy (although we have chosen to write this section in a self-contained manner).

To obtain a convex relaxation of~\eqref{qp}, we linearize it by
replacing each quadratic term $x_{ij}\cdot x_{ij'}$ by a new variable $x_{\{ij,
ij'\}}$ with the exception that $x_{ij}\cdot x_{ij}$ is replaced by the existing variable $x_{ij}$
(since in any binary solution $x_{ij}^2 = x_{ij}$). For notational
convenience, we also refer to variable $x_{ij}$ as $x_{\{ij\}}$ and we 
introduce an auxiliary variable $x_\emptyset$ and set $x_\emptyset = 1$. 
The set of variables of our
relaxation is thus  $\{x_\emptyset\}  \cup \{x_{\{ij\} \cup \{ij'\}}\}_{i\in M,
j,j'\in J}$. Clearly any intended solution satisfies that $\sum_{i\in M}
x_{ij} = 1$ and that $x$ is non-negative. Another family of valid
constraints is as follows.  For a machine $i\in M$, let $X^{(i)}$ be the
$(n+1) \times (n+1)$ matrix whose rows and columns are indexed by $\emptyset$ and
$\{ij\}_{j\in J}$. The entries of $X^{(i)}$ are defined by
$X^{(i)}_{S,T} = x_{S \cup T}$. In particular, this implies that 
$X^{(i)}_{\emptyset,\{ij\}} = X^{(i)}_{\{ij\},\{ij\}} = x_{ij}$ (that we will use crucially).
We impose the constraint that  $X^{(i)}
\succeq 0$. These are valid constraints: indeed, if
$X^{(i)}$ corresponds to an integral assignment $x$ then 
%
\[  X^{(i)} =  z z^T \succeq 0  \qquad \mbox{where} \quad
   z=  (1 , x_{i1},  \cdots, x_{in})^{T} \]
and $X^{(i)}_{\{ij\},\{ij\}} = (zz^T)_{ij,ij} = x_{ij} x_{ij} = x_{ij} = X^{(i)}_{\emptyset,\{ij\}}$.
  

%
\vspace{2mm}
The above yields the following convex (semidefinite programming) relaxation of our problem: 
\begin{center}
  \begin{minipage}{13cm}
\begin{Lalign*}
  \label{sdp}
  \tag{\mbox{SDP}} \qquad
  \textrm{minimize} \qquad \sum_{i\in M} \sum_{j\in J} &w_j  \left( \sum_{j' \in J: j' \preceq_i j}  p_{ij'} x_{\{ij\} \cup \{ij'\}} \right) \\[2mm]
  \textrm{subject to} \qquad  \sum_{i\in M} x_{ij} &= 1 \qquad \mbox{for all } j \in J, \\
  X^{(i)} & \succeq 0 \qquad \mbox{for all } i \in M,\\[1mm]
  x_\emptyset & = 1, \\
  X^{(i)}_{S,T} & \geq 0 \qquad \mbox{for all } i \in M \mbox{ and } S, T \subset J  \mbox{ with  } |S|,|T| \leq 1 .
\end{Lalign*}
\end{minipage}
\end{center}
%

\subsection{Lower bounds on the objective value} 
We briefly sketch why this SDP is stronger; e.g., it is exact on the $3/2$ integrality gap instance from Section \ref{s:prel}.

Similar to previous works, our analysis reduces to that of fixing a single
machine $i$ and analyzing the cost of that machine: we  compare the contribution of that machine to the objective of~\eqref{sdp} to the
(expected) cost of that machine in the schedule returned by our (randomized)
algorithm. To do so, it will be important to understand machine $i$'s contribution  to the objective
 when  a job's processing time
equals its weight, i.e., $p_{ij} = w_j$ for $j\in J$. In this case, 
\begin{align*}
  \sum_{j\in J} &w_j  \left( \sum_{j' \in J: j' \preceq_i j}  p_{ij'} x_{\{ij\} \cup \{ij'\}} \right) =  \sum_{j=1}^n p_{ij} (p_{i1} x_{\{ij\} \cup \{i1\}} + \dots + p_{ij} x_{\{ij\} \cup \{ij\}}), 
\end{align*}
where we numbered the jobs according to the Smith ordering on machine $i$.

Interestingly, we can lower-bound this quantity is various ways as shown in the following lemma. The proof of this lemma crucially uses the SDP constraints and is deferred to the analysis of our approximation guarantee \ifbool{shortVersion}{(in the full version)}{(see Lemma~\ref{lem:genlowerbound})}.
\begin{lemma}
  For any subset $S\subseteq \{1, \dots, n\}$  of jobs,
    \begin{align*}
      \sum_{j=1}^{n} p_{ij}( p_{i1} x_{\{ij\} \cup \{i1\}} +  \dots + p_{ij} x_{\{ij\} \cup \{ij\}}) 
      \geq \sum_{j\not \in S} x_{ij} p_{ij}^2 + \frac{1}{2}\left( \sum_{j\in S}x_{ij} p_{ij}^2 + \left( \sum_{j\in S} x_{ij} p_{ij}\right)^2 \right).
  \end{align*}
  \label{lem:lowerbounds}
\end{lemma}
In particular, we can choose the best set $S$ that gives us the tightest combination of the linear and the quadratic lower bounds. 
In contrast, the relaxations used in~\cite{Skutella01,SethuramanS99} basically take the maximum lower bound (averaged over the machines) obtained by either setting $S= \emptyset$ or $S=J$. 

This flexibility in choosing $S$ will be critical to our analysis. For the $3/2$ gap instance, recall that the linear bound was tight for the large job, while the quadratic bound was tight for the small jobs, which makes the SDP exact on that instance.

\section{Bipartite Assignment with Strong Negative Correlation}
\label{sec:rounding-neg}
As discussed in Section~\ref{s:prel},  independent randomized rounding 
cannot give a better approximation ratio
than $3/2$. To improve upon this ratio, we would ideally like to
introduce strong negative correlation on jobs being assigned to a machine of
the following type: if a job $j$ is assigned to a  machine,  it should be less
likely to assign other jobs to that machine. 
While it is not always impossible to introduce such negative correlations among
all jobs, Theorem~\ref{thm:genassignment}, which we prove in this section, shows
that it is possible to introduce strong negative correlation between subsets of
jobs (or vertices) without introducing positive correlations at pairs of edges
with a common end-point. For convenience, we restate the theorem here.


\begin{reptheorem}{thm:genassignment}
  Let $\zeta = \nicefrac{1}{108}$.
  Consider a bipartite graph $G=(U \cup V, E)$ and let $y \in [0,1]^E$ be
  fractional values on the edges satisfying $ y(\delta(v)) = 1$ for all $v\in
  V$. For each  vertex $u\in U$, select any family of disjoint $E_u^{(1)},
  E_u^{(2)}, \ldots, E_u^{(\kappa_u)} \subseteq \delta(u)$ subsets of edges
  incident to $u$ such that $y(E_u^{(\ell)}) \leq 1$ for $\ell =1 ,\ldots,
  \kappa_u$. Then there exists a   randomized  polynomial-time algorithm
  that outputs a random subset of the edges $E^*\subseteq E$ satisfying
  \begin{enumerate}[(a)]
    \item  For every $v\in V$, we have $|E^* \cap \delta(v)| = 1$ with probability $1$;
    \item  For every $e \in E$, $\Pr[e\in E^*] = y_e$;
    \item  For every $u\in U$ and all $e\neq e'\in \delta(u)$,
      \begin{align*}
        \Pr[ e\in E^* \wedge e' \in E^*] \leq
        \begin{cases}
          (1 - \zeta) \cdot y_ey_{e'} & \mbox{if } e,e' \in E_u^{(\ell)} \mbox{ for some }  \ell \in \{1,2, \ldots, \kappa_u\},\\
          y_e y_{e'} & \mbox{otherwise}.
        \end{cases}
      \end{align*}
  \end{enumerate}
\end{reptheorem}

We start by describing the randomized algorithm  and then give its analysis.


\smallskip \noindent \textbf{Notation: Floating values.} A value $z \in [0,1]$ will be called ``floating" if $z \in (0,1)$. 
\subsection{Algorithm}
We divide the algorithm into three phases and present each phase along with some
simple observations that will be useful in the analysis. \ifbool{shortVersion}{The proofs of these observations and the full analysis of the algorithm can be found in the full version of the paper.}

\paragraph{Phase 1 (Forming the collection $R^*$):} Let $y^*$ denote the initial fractional assignment.
For each vertex $v\in V$, partition its incident edges
$\delta(v)$ into at most $6$ disjoint groups by letting each group --except possibly for at most one group -- be a minimal set of incident edges whose $y^*$-values sum up to at least $1/6$.  (Note that this results in at most $6$ groups  since $y^*(\delta(v)) = 1$, and that these groups can be formed arbitrarily by picking the edges in $\delta(v)$ greedily in non-increasing order of $y^*$-value; the last group may have $y^*$-value smaller than $1/6$.)
Now select a random group, uniformly at random and independently for each vertex $v$, and let $R^*$ be the set of
selected edges. 

\begin{observation}
  \label{obs:p1}
  Let $e, e' \in \delta(u)$ for some $u\in U$. Then, $\Pr[(e\in R^*) \wedge (e' \in R^*)] \geq 1/36$.
\end{observation}
\ifbool{shortVersion}{}{\begin{proof}
  The events that $e\in R^*$ and that $e' \in R^*$ are independent as they both are
  incident to different vertices in $V$. Now the statement follows as 
  each $v\in V$ selects a random group out of at most $6$ many.
\end{proof}}

\begin{figure}[t!]
  \centering
  \begin{tikzpicture}
    \node[vertex] (v1) at (0,0) {\small $v_1$};
    \node[vertex] (v2) at (2,0) {\small $v_2$};
    \node[square] (u1) at (-1, 2) {\small $u_1$};
    \node[square] (u) at (1, 2) {\small $u$};
    \node[square] (u2) at (3,2) {\small $u_2$};
    \draw (v1) edge[dashed]  (u1) edge[ultra thick] (u);
    \draw (v2) edge (u) edge[dashed, ultra thick] (u2);
  \end{tikzpicture}
  \caption{Illustration of the update in phase 2. Solid edges are in $R$ and either (i) thick edges are increased by $\alpha$ and slim edges are decreased by $\alpha$ or (ii) slim edges are increased by $\beta$ and thick edges are decreased by $\beta$. We note that $u_1$ may equal $u_2$ but they  both differ from $u$.}
  \label{fig:alg}
\end{figure}
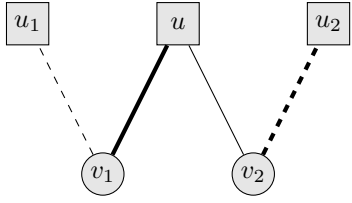

\paragraph{Phase 2 (Updating the assignment):} Initially let $y= y^*, R=R^*$. Repeat the following steps while there exist edges $\{u,v_1\}, \{u,v_2\} \in R \cap E_u^{(\ell)}$ for some $\ell$ and $\{u_1, v_1\} \in \delta(v_1) \setminus R$ and $\{u_2, v_2\} \in \delta(v_2) \setminus R$ with \emph{floating} $y$-value. Here $u, u_1, u_2 \in U$, $v_1, v_2 \in V$, but are otherwise arbitrary. See Figure~\ref{fig:alg}: 
\begin{enumerate}
  \item Let 
      $\alpha  = \min \{y_{u_1, v_1}, 1- y_{u, v_1}, y_{u,v_2}, 1-y_{u_2, v_2}\}\mbox{ and } 
      \beta  = \min \{1-y_{u_1, v_1}, y_{u, v_1}, 1-y_{u,v_2}, y_{u_2, v_2}\}.
    $
  \item \label{y:update} With probability $\frac{\alpha}{\beta + \alpha}$, update $y$ as follows for each $e\in E$: 
    \begin{align*}
      y_e = 
      \begin{cases}
        y_e + \beta & \mbox{if } e= \{u_1, v_1\} \mbox{ or } e = \{u, v_2\}, \\
        y_e - \beta & \mbox{if } e= \{u, v_1\} \mbox{ or } e = \{u_2, v_2\}, \\
        y_e & \mbox{otherwise.}
      \end{cases}
    \end{align*}
    Otherwise (with remaining probability $\frac{\beta}{\alpha+\beta}$), update $y$ as follows for each $e\in E$:
    \begin{align*}
      y_e = 
      \begin{cases}
        y_e - \alpha & \mbox{if } e= \{u_1, v_1\} \mbox{ or } e = \{u, v_2\}, \\
        y_e + \alpha & \mbox{if } e= \{u, v_1\} \mbox{ or } e = \{u_2, v_2\}, \\
        y_e & \mbox{otherwise.}
      \end{cases}
    \end{align*}
  \item \label{reduce:r} For $v\in \{v_1, v_2\}$, if $\sum_{e\in \delta(v) \cap R} y_e = 1$, i.e.~if all the edges incident to $v$ are in $R$, then update $R$ as
    \[R =  (R \setminus \delta(v) )\cup \left\{\arg \max_{e\in \delta(v) \cap R} y_e \right\}.\]
		That is, remove all edges incident to $v$ from $R$, except one with the largest $y$-value. 
\end{enumerate}
We note the following simple observations about this phase.
\begin{observation}
  \label{obs:termination}
  During Phase 2, if a variable $y_e$ reaches $0$ or $1$, then it is not
  updated anymore.  Moreover, at each iteration of Phase 2, at least one edge
  with floating $y$-value has its $y$-value reach $0$ or $1$.
\end{observation}
\ifbool{shortVersion}{}{\begin{proof}
  This follows from that Phase $2$ only updates floating $y$-values and, in each
  iteration, $\alpha$ and $\beta$ is selected so that one of the selected edges' $y$-value
  will reach $0$ or $1$.
\end{proof}}

\begin{observation}
  \label{obs:p2inv}
  Phase~$2$ satisfies the invariants $y(\delta(v)) = 1$ for every $v\in V$
  and $y_e \geq 0$ for every $e\in E$.
\end{observation}
\ifbool{shortVersion}{}{\begin{proof}
  Notice that when $y$ is updated then the selection of $\alpha$ and $\beta$
  guarantees that $y_e \geq 0$ for every $e\in E$. Moreover, the update is
  designed so that the fractional degree of a vertex in $V$ stays constant.
  Thus the statement follows since we start with $y=y^*$ for which
  $y(\delta(v))= 1$ for $v\in V$. 
\end{proof}}

\begin{observation}
  \label{obs:p2R}
	The set $R$ does not increase in size during Phase 2. Moreover, if an edge $e \in \delta(v)\cap R$ is removed from $R$ (in Step \ref{reduce:r}) then it must be that $y(e) \leq 1/2$ after Step \ref{y:update}. 
\end{observation}
\ifbool{shortVersion}{}{\begin{proof}That $R$ only decreases in size follows directly  from Step \ref{reduce:r}. For the second part, if Step \ref{reduce:r} is applied at $v$ and $y(e)>1/2$ for some $e\in \delta(v)$, then as $\sum_{e' \in \delta(v)} y(e')=1$, it must be that $e = \arg \max_{e' \in \delta(v)}$ and thus $e$ remains in $R$.  
\end{proof}}

\begin{observation}
  \label{obs:p2term}
  When Phase~$2$ terminates, then for every $u\in U$ and $\ell \in \{1, \ldots, \kappa_u\}$, we have $|\{e\in E_u^{(\ell)}\cap R\mid y_e > 0\}| \leq 1$.
\end{observation}
\ifbool{shortVersion}{}{\begin{proof}
  Suppose that there exist $e_1, e_2 \in E_u^{(\ell)}\cap R$ with $y_{e_1},
  y_{e_2} > 0$. Then since any iteration of Phase~$2$ maintains the value of
  $y(E_u^{(\ell)} \cap R)$  and $R \subseteq R^*$ we have $y(E_u^{(\ell)} \cap
  R) \leq y^*(E_u^{(\ell)} \cap R^*) \leq 1$. Hence, $y_{e_1}, y_{e_2} < 1$. Now
  by Step~$3$ of Phase~$2$, we are guaranteed that a not-yet-integrally-assigned
  vertex $v\in V$ has $y(\delta(v)\cap R) < 1$. Therefore, there exist edges 
  $e_1 = \{v_1, u\}, e_2 = \{v_2, u\}$ and $\{u_1, v_1\} \in \delta(v_1)
  \setminus R$ and $\{u_2, v_2\} \in \delta(v_2) \setminus R$ with floating
  $y$-values. This implies that Phase~$2$ does not terminate in this case. 
\end{proof}}

\paragraph{Phase 3 (Randomized Rounding):} Form $E^*$ by, independently for each vertex $v\in V$, selecting a single
edge $e\in \delta(v)$  so that $e \in \delta(v)$ is selected with probability
$y_e$.  Notice that this is possible because, by Observation~\ref{obs:p2inv},
we have $ \sum_{e\in \delta(v)} y_e = 1$ for all $v\in V$ and $y_e \geq 0$
for all $e\in E$.

\ifbool{shortVersion}{
  \subsection{Sketch of Analysis}
}
{
  \subsection{Analysis}
}

We first note that the algorithm terminates in polynomial time.
Phase~$1$ and Phase~$3$ both clearly run in polynomial time.
Each step of Phase~$2$ runs in polynomial time and by
Observation~\ref{obs:termination}, Phase~$2$  runs in
at most $|E|$ iterations.

We continue to analyze the properties. The intuition for why they should hold is as
follows. The algorithm is inspired by randomized-rounding algorithms for
bipartite matchings such as  pipage rounding and swap rounding. It is easy to see that these
algorithms satisfy both
Property~\eqref{genprop1} and the marginal probabilities
(Property~\eqref{genprop2}): indeed, $\alpha$ and $\beta$ are defined in order to do so.  Moreover, the weak bound of
Property~\eqref{genprop3} follows basically from the fact that, for each $u\in
U$, the $y$-values of two edges incident to $u$ are never  increased
simultaneously. Finally, the intuition behind the novel strong bound of
Property~\eqref{genprop3} is as follows. After Phase $2$, the probability that
two vertices $e, e' \in E_u^{(\ell)}$ are in $R$ is at least $1/36$. Now using
that the initial $y$-value of edges in $\delta(v)\cap R$ is at most $1/3$ for
every $v\in V$, and that the $y$-values of edges are preserved in expectation, there is a reasonable probability that both $e, e'$ will remain in $R$ until the end.
However, in that case, it is easy to see by Observation~\ref{obs:p2term} that at most one
of them will be selected in $E^*$. 
\ifbool{shortVersion}{
  The formal analysis is deferred to the full version due to space constraints.
}
{
We now continue to formally prove these properties.

\vspace{2mm}
\noindent{\bf Property~\eqref{genprop1}:} That Property~\eqref{genprop1} of
Theorem~\ref{thm:genassignment} holds follows from Observation \ref{obs:p2inv} and as Phase $3$ 
chooses exactly one edge incident to each $v\in V$.

\vspace{2mm}
\noindent{\bf Properties~\eqref{genprop2} and~\eqref{genprop3}:} 
To show these properties, we will inductively show some invariants. Let $Y^{(k)} = (y_e^{(k)}:~e \in E)$ denote the collection of $y$-values of edges and $R^{(k)}$ be the set $R$ at the end of iteration $k$ of Phase 2. 
For an edge $e = \{u,v\} \in R$ with $u\in U$ and $v\in V$  let $R_{\bar e} = \{e'\in \delta(v) \cap R: e' \neq e\}$ be the other edges in $R$ incident to $v$.

 We show the following invariants hold after each iteration $k$. Here, conditioning an event on $Y^{(k)}$ and $R^{(k)}$ means  the probability of that event if the random iterations in Phase 2 are applied starting from the assignment $Y^{(k)}$ and $R=R^{(k)}$.


\begin{align}
  \label{eq:IH1} \Pr[e\in E^* \bigm| Y^{(k)},R^{(k)} ] & = y^{(k)}_e  \qquad \forall e \in E \\
  \label{eq:IH2} \Pr[e \in E^* \wedge e' \in E^* \bigm| Y^{(k)},R^{(k)} ] & \leq y^{(k)}_e y^{(k)}_{e'} \qquad \forall u\in U, e,e' \in \delta(u) \\
  \label{eq:IH3} \Pr[e \in E^* \wedge e' \in E^*  \bigm| Y^{(k)},R^{(k)} ] & \leq 2(y^{(k)}(R^{(k)}_{\bar e}) + y^{(k)}(R^{(k)}_{\bar e'})) y^{(k)}_e y^{(k)}_{e'}  \\
	 & \forall u\in U, \ell \in \{1, \ldots, \kappa_u\}, e\neq e'\in E_u^{(\ell)} \cap R^{(k)} \nonumber
\end{align}


To show these, we will apply reverse induction. For the base case, we show that these properties hold after the last iteration of Phase 2. For the inductive step, we show that if they hold after the $k$-th iteration then they also hold after iteration $k-1$ (or equivalently at the beginning of iteration $k$), and hence they also hold for the $y$-values and $R$ at the beginning of Phase 2.

Let us first see how this implies the theorem.

At the beginning of Phase $2$ we have $y^{(0)}=y^*$ and $R^{(0)}= R^*$. So having~\eqref{eq:IH1} for $k=0$, implies Property~\eqref{genprop2} and~\eqref{eq:IH2} implies the weaker bound in Property~\eqref{genprop3}. For the stronger bound, consider two edges $e\neq e' \in E_u^{(\ell)}$.  By \eqref{eq:IH3} above, we have that
\begin{align*}
  \Pr[e \in E^* \wedge e' \in E^*] & = \E_{R^*} [ \Pr[(e \in E^* \wedge e' \in E^*) \bigm| Y^*,R^*]] \\
	&\leq \Pr[e,e' \in R^*] \cdot  2(y^*(R_{\bar e}) + y^*(R_{\bar e'})) y^*_e y^*_{e'} + (1- \Pr[e,e' \in R^*]) \cdot  y^*_e y^*_{e'} \\
  &\leq \Pr[e,e' \in R^*]  \frac{2y^*_e y^*_{e'}}{3} + (1- \Pr[e,e' \in R^*]) y^*_e y^*_{e'} \\
  &\leq   \frac{2y^*_e y^*_{e'}}{3\cdot 36} + \frac{35}{36} y^*_e y^*_{e'} \\
  & =  \frac{107}{108} y^*_e y^*_{e'}.
\end{align*}
The second inequality follows from the fact that $y^*(R_{\bar e}), y^*(R_{\bar
e'}) \leq 1/6$ because $ e, e' \in R$ and after Phase~$1$ 
$R\cap \delta(v)$ is a \emph{minimal} group with $y^*$-value at least $1/6$ for each $v\in V$; and the third inequality follows from Observation~\ref{obs:p1}.

It thus remains to prove~\eqref{eq:IH1}-\eqref{eq:IH3} by reverse induction on the iterations in Phase 2. 
One subtle point in the argument is that the set $R$ might also change (reduce in size) after an iteration.

\medskip
\noindent {\bf Base case (when Phase 2 terminates):} 
    In this case Phase $2$ will not change any of the $y$-values. As each vertex $v \in V$ picks
		an edge in $\delta(v)$ randomly with probability $y_e$,  $\Pr[
    e\in E^*]= y_e$ for every $e\in E$, so~\eqref{eq:IH1} is
    satisfied. 	
	Similarly for~\eqref{eq:IH2}, we note that for two edges $e\neq e' \in \delta(u)$, it holds that
    $\Pr[e\in E^* \wedge e' \in E^*] = y_e y_{e'}$.
		
		Finally, 
    Observation~\ref{obs:p2term} says that the number of edges in $E^{(\ell)}_u \cap
    R$ with positive $y$-value is at most $1$. Therefore, we have that
    $\Pr[e\in E^* \wedge e' \in E^*] = 0$ for $e\neq e' \in E^{(\ell)}_u \cap R$
    and~\eqref{eq:IH3} holds trivially.

\medskip
\noindent
   {\bf Inductive step:}  Assuming~\eqref{eq:IH1}-\eqref{eq:IH3} holds at the end of iteration $k$, we prove that they hold at the end of
	iteration $k-1$).
			
For notational ease, let us denote $Y = Y^{(k-1)}$, $R = R^{(k-1)}$ and let  $Y' = Y^{k}$ and  $R'=R^{(k)}$ denote the (random) updated $y$-values and set $R$.

		
		
    
    We first verify~\eqref{eq:IH1}. By the inductive hypothesis (I.H.) we have that
    $\Pr[e \in E^* \bigm| Y'] = y'_e$. So, 
		$\Pr[e \in E^* \bigm| Y] = \E_{Y'|Y} [ \Pr[e \in E^* \bigm| Y']]$, which is  $\E_{Y'|Y} [y'_e]$.
		If Phase $2$ did not update the value of edge
    $e$ then clearly $y'_e = y_e$. Otherwise, we have  $\E_{Y'|Y}[y'_e]
    = \frac{\alpha}{\alpha+\beta} (y_e + \beta)
    + \frac{\beta}{\alpha+\beta}(y_e - \alpha) = y_e$.  Thus,~\eqref{eq:IH1} holds in either case.

    Similarly, we show~\eqref{eq:IH2}.
		By the I.H., $\Pr[e\in E^* \wedge e'\in E^*| Y'] \leq y'_e y'_{e'}$
			 and thus
		$$ 	\Pr[e\in E^* \wedge e'\in E^* \bigm| Y]  = \E_{Y'|Y}[ \Pr[e\in E^* \wedge e'\in E^* \bigm| Y'] \leq  \E_{Y'|Y}[y'_e y'_{e'}].$$
		On the one hand, if Phase $2$ only
    changed the $y$-value for at most one of $e$ and $e'$, then by independence
    this is at most $\E_{Y'|Y}[y'_e] \E_{Y'|Y}[y'_{e'}] = y_e y_{e'}$.  On the other hand, if it changed both of
    the values then we have
    $$\E_{Y'|Y} [y'_e y'_{e'}]  = \frac{\alpha}{\alpha+\beta} (y_e + \beta)(y_{e'}
    - \beta) + \frac{\beta}{\alpha+\beta} (y_e - \alpha) (y_{e'} + \alpha) \leq
    y_e y_{e'}.$$
		Indeed, if Phase $2$ changes the value of two edges incident
    to a vertex in $U$ then it always increases the value of one edge and
    decreases the value of the other edge.  We have thus that~\eqref{eq:IH2} is
    satisfied.

    We finish the analysis by verifying~\eqref{eq:IH3}.  Consider $e\neq e' \in
    E_u^{(\ell)} \cap R$ for some $u\in U$ and $\ell \in \{1, \ldots,
    \kappa_u\}$. We wish to show that 
				  \begin{align*}
      \Pr[e\in E^* \wedge e' \in E^* \bigm| Y] \leq 2(y(R_{\bar e}) + y(R_{\bar e'})) y_e y_{e'}. 
    \end{align*}
		
		Let $R'$ be the set $R$ after the single iteration of Phase~$2$.
		As previously we will use that
		\[ \Pr[e\in E^* \wedge e' \in E^* \bigm| Y] = \E_{Y'|Y} [\Pr[e\in E^* \wedge e' \in E^* \bigm| Y']\] 
	 	and the I.H., but we cannot do it directly as it might the case that even though $e$ and $e'$ belong to $R$, they may not belong to $R'$. 
		So we condition the right hand side depending on whether this happens or not.
		
		Suppose $e \not \in R'$. Then by Observation \ref{obs:p2R}, $y'(R_{\bar e}) \geq 1/2$ and hence we have that $2(y'(R_{\bar e}) + y'(R_{\bar e'})) y'_e y'_{e'} \geq y'_e y'_{e'}$.
		By~\eqref{eq:IH2} we have that 
		\[\Pr[e\in E^* \wedge e' \in E^* \bigm| Y',R'] \leq y'_e y'_{e'}, \]
		this implies that (conditioned on $e \not\in R'$)
		\[\Pr[e\in E^* \wedge e' \in E^* \bigm| Y',R'] \leq 2(y'(R_{\bar e}) + y'(R_{\bar e'})) y'_e y'_{e'}.\]
		The same holds if $e' \not \in R'$. 
		
		Now if both $e$ and $e'$ lie in $R'$ by the I.H. we know that 
		    \begin{align*}
      \Pr[e\in E^* \wedge e' \in E^* \bigm| Y',R'] \leq 2(y'(R'_{\bar e}) + y'(R'_{\bar e'})) y'_e y'_{e'} \leq 2(y'(R_{\bar e}) + y'(R_{\bar e'})) y'_e y'_{e'},
    \end{align*}
    where the last inequality follows from Observation~\ref{obs:p2R}, i.e., from the fact that $R' \subseteq R$.
	
    We have thus upper bounded all cases (irrespective of whether $R'$ contains
    $e$ or $e'$) by the same expression and it suffices to show
		 \begin{align*}
		\E_{Y'|Y} [2(y'(R_{\bar e}) + y'(R_{\bar e'})) y'_e y'_{e'}] \leq  2(y(R_{\bar e}) + y(R_{\bar e'})) y_e y_{e'}
		   \end{align*}
    If neither $e$ or $e'$  is changed by the iteration of Phase~$2$, then
    \begin{align*}
      \E_{Y'|Y}[2(y'(R_{\bar e}) + y'(R_{\bar e'})) y'_e y'_{e'}] & = \E_{Y'|Y}[2(y'(R_{\bar e}) + y'(R_{\bar e'}))] y_e y_{e'} = 2(y'(R_{\bar e}) + y'(R_{\bar e'})) y_e y_{e'},
    \end{align*}
    where the second equality follows by linearity of expectation and~\eqref{eq:IH1}.

    Now suppose the iteration of Phase~$2$ changes at least one of $e$ or $e'$.
    Then we claim that $y'(R_{\bar e})= y(R_{\bar e})$ and $y'(R_{\bar e'})
    = y(R_{\bar e'})$. To see this note that an iteration of Phase~$2$ changes
    exactly two edges in $R$ incident to the same vertex in $U$ and since, in this case, one of them is incident to $u$ so must
    the other one. The sets $R_{\bar e}$ and $R_{\bar e'}$ only contain edges
    of $R$ that are not incident to $u$ and are thus left unchanged in this case. Hence,
    \begin{align*}
      \E_{Y'|Y}[2(y'(R_{\bar e}) + y'(R_{\bar e'})) y'_e y'_{e'}] & = 2(y(R_{\bar e}) + y(R_{\bar e'})) \E_{Y'|Y}[y'_e y'_{e'}] \\
    & \leq  2(y(R_{\bar e}) + y(R_{\bar e'})) y_e y_{e'},
    \end{align*}
    where the last inequality follows from~\eqref{eq:IH2}. We have thus also
    proved~\eqref{eq:IH3} which completes the proof.

}

\section{Rounding the Fractional Schedule}
\label{s:alg}

We now describe our scheduling algorithm. The algorithm solves the SDP
relaxation  from Section~\ref{sec:relaxation}, and applies the bipartite
rounding procedure from Section~\ref{sec:rounding-neg} to a suitably defined
graph based on the SDP solution. 
We will analyze this algorithm in Section \ref{s:alg-analysis}, and 
in particular show the following result which directly implies Theorem \ref{thm:intro-main}.

\begin{theorem}
  The expected cost of the rounding algorithm is at most 
  $(\nicefrac{3}{2} - c)$ times the cost of the optimal solution to the relaxation, where
  $c = \nicefrac{\zeta}{20000}$ and $\zeta$ is the constant in Theorem~\ref{thm:genassignment}.
  \label{thm:main2}
\end{theorem}

\subsection{Description of Algorithm}
\label{s:alg-desc}

Our rounding algorithm consists of defining groups (i.e., the families $E_u^{(\ell)}$) for each machine and then
applying Theorem~\ref{thm:genassignment}. Specifically, let $x$ denote an
optimal solution to our relaxation. We shall interpret the vector  $y
= (x_{ij})_{i\in M, j\in J}$ as an fractional assignment of jobs to machines
in the bipartite graph $G = (M \cup J, E)$ where $E = \{ ij: y_{ij} >0\}$.
Notice, that $y(\delta(j)) = 1$ for each $j\in J$ and $y \geq 0$. Thus,
$y$ satisfies the assumptions of Theorem~\ref{thm:genassignment}. It remains
to partition the edges incident to the machines into groups. To do this, we apply the following grouping 
procedure to each machine separately. 

\vspace{2mm}
\noindent{\bf Grouping Procedure:} For a fixed machine $i$ we define the groups as follows:
\begin{enumerate}
  \item Call a job $j$ of class $k$, if $p_{ij} \in [10^{k-1},10^{k})$. We assume (by scaling) that $p_{ij}\geq 1$ if $p_{ij} \not= 0$.
  \item  For each class $k=0,1,2,\ldots$, order the jobs in that class in non-increasing  order of Smith's
    ratio, i.e., $w_j/p_{ij}$, and form groups as follows.  If some job $j$ has
    $x_{ij} \geq 1/10$, it forms a separate group by itself $\{j\}$. For the
    remaining jobs, greedily pick the jobs in class $k$ so that their  total
    fractional $y$-value on $i$ first reaches at least $1/10$ and make it
    a group; and repeat until the remaining jobs of that class have total fractional value less than $1/10$.
\end{enumerate}

By definition, the ungrouped jobs in each size class $k$ have total fractional value
less than $1/10$ on machine $i$.  Note also that several singleton groups could
be interspersed between jobs of a single group. For an example see
Figure~\ref{fig:grouping}.

Let $E_i^{(1)},\ldots,E_i^{(\kappa_i)}$ denote the groups formed, over all the classes, for machine $i$. 
We now apply Theorem~\ref{thm:genassignment} to the graph $G = (M\cup J,E)$ with $U=M$ and the groups $E_u^{(1)},\ldots,E_u^{(\kappa(u))}$ at the machine $u\in U$.  Observe that the conditions of the groups are
satisfied, i.e.,~they are disjoint and the total $y$-value is at most $1$ in
all of them.
This gives an assignment of the jobs to machines and thus a schedule. 

\begin{figure}[t]
  \begin{center}
    \begin{tikzpicture}
      \draw[fill=black!20!white] (0,0) rectangle (1,1.8);
      \node at (0.5, -0.3) {\small $\nicefrac{1}{12}$};
      \draw[fill=black!50!white](1,0) rectangle (3,2.2);
      \node at (2, -0.3) {\small$\nicefrac{1}{6}$};
      \draw[fill=black!20!white] (3,0) rectangle (4,1);
      \node at (3.5, -0.3) {\small$\nicefrac{1}{12}$};
      \draw[fill=black!70!white](4,0) rectangle (4.7,1.5);
      \node at (4.35, -0.3) {\small $\nicefrac{1}{18}$};
      \draw[fill=black!70!white]  (4.7,0) rectangle (5.7,1);
      \node at (5.2, -0.3) {\small$\nicefrac{1}{12}$};
      \draw (5.7,0) rectangle (6.7,1.2);
      \node at (6.2, -0.3) {\small$\nicefrac{1}{12}$};
      \draw [decorate,decoration={brace,amplitude=10pt,mirror}]
      (0,-0.6) -- (6.7,-0.6) node [black,midway,yshift=-1.0cm] 
      {\footnotesize \begin{minipage}{9.5cm}The jobs are ordered in non-increasing order of Smith's ratio and the widths of the depicted jobs show their $y$-value on the considered machine.\end{minipage}};
      \draw [decorate,decoration={brace,amplitude=10pt}]
      (-0.3, 0) -- (-0.3, 2.2) node [black,midway,xshift=-3.0cm] 
      {\footnotesize \begin{minipage}{4.5cm}The height of the jobs represent their processing times which are all in $[10^{k-1}, 10^{k})$ since we only consider jobs of class $k$.\end{minipage}};
          \end{tikzpicture}
          \caption{Example of the grouping procedure on a machine $i$ for the jobs of class $k$. The different groups are depicted in different colors; the job corresponding to the white rectangle is ungrouped.}
    
    \label{fig:grouping}
  \end{center}
\end{figure}


\ifbool{shortVersion}{
  \subsection{Sketch of Analysis}
  \label{s:alg-analysis}
  Due to space constraints, the analysis is deferred to the full version.
  Instead, we give some high level intuition here.  The analysis is reduced to
  comparing the (expected) cost of a fixed machine $i$ in the constructed
  schedule compared to the contribution of that machine to the objective of our
  relaxation~\eqref{sdp}. To compare these two quantities for machine $i$,
  there are two cases. 
  
  In the first case, machine $i$ has a substantial contribution due to ``grouped'' jobs.
			Then, the grouping together with the strong
  negative correlation of Theorem~\ref{thm:genassignment} yields
  a significantly better solution than independent randomized rounding which in
  turn implies an approximation guarantee of $\nicefrac{3}{2}-c$. 

  In the other case, most of the contribution is due to ungrouped jobs, i.e.,~the jobs are very
  different in terms  of size and most classes have at most $1/10$ fraction of jobs.
	In this case,  there may be no gain in
  using our rounding algorithm compared to independent randomized rounding.
  Instead, we show that independent
  randomized rounding performs better than $\nicefrac{3}{2} - c$ in this case.
  Intuitively, this holds simply because when jobs are very different then they
  will not affect each others' completion times significantly. Note that, in the
  bad example of independent rounding in Section~\ref{s:prel},
each job class has a ``large'' (actually one unit of) fractional assignment to
a machine.

}
{

\subsection{Analysis}
\label{s:alg-analysis}
To analyze the performance of the algorithm above, we proceed in several steps. We first define some notation and
make some observations that allow us to express the cost of the algorithm and the relaxation in a more convenient form.
In Section \ref{sec:upperbound} we show how to upper bound the cost of the schedule produced by the algorithm.
In secton \ref{sec:lowerbound} we show how to derive various strong lower bounds from the SDP formulation, and finally in Section \ref{sec:finalproof} we show how to combine these results to obtain Theorem \ref{thm:main2}.

\subsubsection{Notation}
Let $X_{ij}$ denote the random indicator variable that takes value $1$ if the
algorithm assigns job $j$ to machine $i$.
The expected value of the returned schedule of the algorithm can then be
written as $\sum_{i \in M} \oalg_i$, where $\oalg_i$ denotes the expected cost of machine $i$, i.e., 
\begin{align*}
  \oalg_i = \E \left[\sum_{j\in J} X_{ij} w_j \left( \sum_{j' \preceq j} X_{ij'} p_{ij'} \right)  \right]
  = \sum_{j\in J}  w_j \left( \sum_{j' \preceq j}  p_{ij'}\E[X_{ij}X_{ij'}] \right).
\end{align*}
Similarly, the value of the optimal solution $x$ to the relaxation can be
decomposed into a sum  $\sum_{i\in M} \orel_i$ over the costs of the machines,
where 
\begin{align*}
   \orel_i = \sum_{j\in J} w_j \left( \sum_{j'\in J: j' \preceq j} p_{ij'} x_{\{ij\} \cup \{ij'\}} \right).
\end{align*}
In order to prove Theorem~\ref{thm:main2}, it is thus sufficient to show
\begin{align}
  \oalg_i \leq \left( 3/2 - c \right)\orel_i \qquad \mbox{for all } i\in M.
  \label{eq:toprove}
\end{align}
To this end, we fix an arbitrary machine $i\in M$ and use the following notation:
\begin{itemize}
  \item For simplicity, we abbreviate $p_{ij}$ by $p_j$, $x_{ij}$ by $x_{j}$, $x_{\{ij\}
    \cup \{ij'\}}$ by $x_{\{j\} \cup \{j'\}}$, and  $X_{ij}$ by $X_j$. 
  \item We let $\beta_j = w_j/p_j$ denote Smith's ratio of  job $j\in J$ on
    machine $i$ and rename the  jobs $J= \{1, 2, \ldots, n\}$ so that
    $\beta_1 \leq \beta_2 \leq \dots \leq \beta_{n}$.
\end{itemize}
With this notation, we can rewrite $\orel_i$ and $\oalg_i$ as follows.
\begin{lemma}
  We have
  \begin{align*}
    \oalg_i & =  \sum^n_{j=1} (\beta_j -\beta_{j+1}) \E \left[\sum_{j'=1}^j p_{j'} X_{j'} (p_1 X_1 + \dots + p_{j'} X_{j'})  \right], \\
    \orel_i & =\sum_{j=1}^n (\beta_j - \beta_{j+1}) \left[ \sum_{j'=1}^j p_{j'}( p_1 x_{\{j'\} \cup \{1\}} +  \dots + p_{j'} x_{\{j'\} \cup \{j'\}}) \right], 
  \end{align*}
  where for notational convenience we let $\beta_{n+1} = 0$.
  \label{lem:rewrite}
\end{lemma}
\begin{proof}
  We prove the first equality based on a telescoping sum argument.
	The second equality follows exactly by the same arguments. Using $w_j = \beta_j p_j$ we can rewrite
  \begin{align*}
    \oalg_i & = \E \left[ \sum^n_{j=1} w_j \left( \sum_{j'=1}^j X_{j} X_{j'} p_{j'} \right) \right]  
     =  \E \left[ \sum^n_{j=1} \beta_j p_j X_j \left( \sum_{j'=1}^j X_{j'} p_{j'} \right) \right].
  \end{align*}
  We now claim that the right-hand side of this expression equals 
  \begin{align*}
    \sum^n_{j=1} (\beta_j -\beta_{j+1}) \E \left[\sum_{j'=1}^j p_{j'} X_{j'} (p_1 X_1 + \dots + p_{j'} X_{j'})  \right].
  \end{align*}
  Consider any term $p_k X_k p_\ell X_\ell$ with
  $k\leq \ell$. This term appears in $\E \left[ \sum^n_{j=1} \beta_j p_j X_j
  \left( \sum_{j'=1}^j X_{j'} p_{j'} \right) \right]$ only  when $j'= k $ and
  $j= \ell$ and has a coefficient of $\beta_\ell$. The same term appears in the expression
  $\sum_{j=1}^n (\beta_j - \beta_{j+1}) \E\left[ \sum_{j'=1}^j p_{j'} X_{j'} (p_1
    X_1 + \dots + p_{j'} X_{j'}) \right]$ when $j = \ell, \ell+1, \dots, n$ with
    coefficients $(\beta_\ell - \beta_{\ell+1}), (\beta_{\ell+1} - \beta_{\ell+2}),
    \dots, (\beta_n - \beta_{n+1})$.  Thus, by telescoping, the coefficient in
    front of $p_k X_k p_\ell X_\ell$ is again $\beta_\ell$. 
\end{proof}
By combining the above lemma with~\eqref{eq:toprove}, we have further reduced our task of proving Theorem~\ref{thm:main2} to that of proving
\begin{align}
  \E \left[\sum_{j=1}^{n'} p_{j} X_{j} (p_1 X_1 + \dots + p_{j} X_{j})  \right] \leq (3/2-c)\left[ \sum_{j=1}^{n'} p_{j}( p_1 x_{\{j\} \cup \{1\}} +  \dots + p_{j} x_{\{j\} \cup \{j\}}) \right].
  \label{eq:toprove2}
\end{align}
for all $n'\in J$. The rest of this section is devoted to proving this
inequality for a fixed $n'$. We shall use the following notation:
\begin{itemize}
  \item Let $G$ denote those jobs that are in the groups that only contain jobs
    from $\{1, \dots, n'\}$. Let  $\ov{G} = \{1, \dots, n'\} \setminus G$
    denote the ``ungrouped'' jobs. Note that, by the definition of the
    algorithm, specifically, the grouping, we have that each job class has
    fractional value less than $1/10$ in $\ov{G}$. Let 
		$\cG$ denote the collection of these groups restricted to jobs $\{1, \dots, n'\}$. 
  \item Let $L= \sum_{j=1}^{n'} x_j p_j$ denote the ``linear'' sum and let $Q
    = \sum_{j=1}^{n'} x_j p_j^2$ denote the ``quadratic'' sum. We also use the
    notation $\ov{L}$ and $\ov{Q}$ to denote the linear and quadratic sums when
    restricted to ungrouped jobs, i.e., $\ov{L} = \sum_{j\in \ov{G}} x_j p_j$
    and $\ov{Q} = \sum_{j\in \ov{G}} x_j p_j^2$.
\end{itemize} 
The proof of~\eqref{eq:toprove2} is described over the following three subsections.
In Section~\ref{sec:upperbound} we give an upper bound on the left-hand-side
(LHS) of~\eqref{eq:toprove2}; in Section~\ref{sec:lowerbound} we give several
lower bounds on the right-hand-side (RHS) of~\eqref{eq:toprove2}; finally, in
Section~\ref{sec:finalproof} we combine these bounds to
prove~\eqref{eq:toprove2}.

\subsubsection{Upper bound on the LHS of~\eqref{eq:toprove2}}
\label{sec:upperbound}

We give the following upper bound on the LHS of~\eqref{eq:toprove2}.
The lemma essentially say that we have a ``gain'' of $O(\zeta)$ for each grouped job, which
follows from our negative correlation rounding. 
\begin{lemma} For $Q, \ov{Q}$ and $L$ as defined above, 
we have
  \begin{align*}
    \E \left[\sum_{j=1}^{n'} p_{j} X_{j} (p_1 X_1 + \dots + p_{j} X_{j})  \right] 
    \leq (1-\nicefrac{\zeta}{200})\cdot Q + \nicefrac{\zeta}{200}\cdot  \ov{Q} + \nicefrac{1}{2}\cdot L^2.  
  \end{align*}
  \label{lem:upperbound}
\end{lemma}
\begin{proof}
  Using that $X_j^2=X_j$ and a simple recombination of the terms, we have that 
  \begin{align*}
    \E \left[\sum_{j=1}^{n'} p_{j} X_{j} (p_1 X_1 + \dots + p_{j} X_{j})  \right] 
    & = \E \left[ \frac{1}{2} \sum_{j=1}^{n'} X_j p_j^2 +  \frac{1}{2}\left( \sum_{j=1}^{n'}X_j p_j \right)^2 \right]
  \end{align*}
  As our rounding satisfies the marginals, this can be simplified to
  $\frac{1}{2}\sum_{j=1}^{n'} x_j p_j^2 + \frac{1}{2}\E\left[ \left(
  \sum_{j=1}^{n'} X_j p_j \right)^2 \right]$.
  We now upper bound the latter term.
  \begin{align*}
    & \E\left[ \left( \sum_{j=1}^{n'} X_j p_j \right)^2 \right]  =  \E \left[ \sum_{j,j'} X_{j} X_{j'} p_{j} p_{j'} \right] =  \E \left[ \sum_{j,j':j\neq j'} X_{j} X_{j'} p_{j} p_{j'} \right]  + \E \left[ \sum_{j:j =j'} X_{j} X_{j'} p_{j} p_{j'} \right] \\[2mm]
		 \leq & \left(\sum_{j \neq j'} x_{j} x_{j'} p_{j} p_{j'} - \zeta \sum_{G'\in \cG} \sum_{j\neq j' \in G} x_{j} x_{j'} p_{j} p_{j'}\right)  + \sum_j x_j p_j^2 
	\qquad  (\mbox{by Theorem~\ref{thm:genassignment} and } \E[X_j^2] =\E[X_j]=x_j)\\[2mm]
		= & \sum_{j,j'} x_{j} x_{j'} p_{j} p_{j'} + \sum_{j} (x_{j} -x^2_{j}) p_{j}^2 - \zeta \sum_{G'\in \cG} \sum_{j\neq j' \in G'} x_{j} x_{j'} p_{j} p_{j'} \\[2mm]
     \leq &  \left( \sum_{j} x_{j} p_{j} \right)^2  + \sum_{j} x_{j} p_{j}^2  - \zeta \sum_{G'\in \cG} \left( \sum_{j \in G'} x_{j} p_{j} \right)^2  \qquad \mbox{(since $\zeta \leq 1$)}.
  \end{align*}
	
  Now, for each group $G'\in \cG$, we have $\sum_{j \in G'} x_{j} \geq 1/10$
  and $p_{j} \geq p_{j'}/10$ for $j,j' \in G$. Therefore, \[(\sum_{j \in
  G'} x_{j} p_{j})^2 \geq \sum_{j \in G'} x_{j} p_{j}^2 /100.\]
    Thus, we  have that the expected cost of the machine is upper bounded by
  \begin{align*}
    &  \left(\sum_{j = 1}^{n'} x_{j} p_{j}^2 \right) + \frac{1}{2} \left(\left( \sum_{j=1}^{n'} x_{j} p_{j} \right)^2 - \left(\zeta \sum_{G' \in \cG} \sum_{j\in G'} x_{j} p_{j}^2 /100 \right)\right) \\
		& = \left(\sum_{j = 1}^{n'} x_{j} p_{j}^2 \right) + \frac{1}{2} \left(\left( \sum_{j=1}^{n'} x_{j} p_{j} \right)^2 - \left(\zeta  \sum_{j\in G} x_{j} p_{j}^2 /100 \right)\right) \\
    & = (1-\nicefrac{\zeta}{200}) \left(\sum_{j = 1}^{n'} x_{j} p_{j}^2 \right) + \nicefrac{\zeta}{200} \left( \sum_{j\in  \ov{G}} x_j  p_j^2 \right) + \nicefrac{1}{2} \left( \sum_{j=1}^{n'} x_j p_j \right)^2.
  \end{align*}
\end{proof}

\subsubsection{Lower bounds on the RHS of~\eqref{eq:toprove2}}
\label{sec:lowerbound}
The lemma below gives a general lower bound that allows us to the RHS of~\eqref{eq:toprove2} in various ways by choosing different subsets $S$.
The particular lower bounds that we
later use (by plugging particular choices of $S$) are then stated in Corollary \ref{cor:bounds}.
\begin{lemma}
  For any subset $S\subseteq \{1, \dots, n'\}$  of jobs\footnote{Here, and in
    the following, we mean  $j\in \{1, \dots, n'\} \setminus S$ by $j\not \in
  S$.},
    \begin{align*}
      \sum_{j=1}^{n'} p_{j}( p_1 x_{\{j\} \cup \{1\}} +  \dots + p_{j} x_{\{j\} \cup \{j\}}) 
     \geq \sum_{j\not \in S} x_j p_j^2 + \frac{1}{2}\left( \sum_{j\in S}x_j p_j^2 + \left( \sum_{j\in S} x_j p_j\right)^2 \right).
  \end{align*}
  \label{lem:genlowerbound}
\end{lemma}
\begin{proof}
  Similar to the proof of Lemma~\ref{lem:upperbound}, 
  \begin{align*}
    \sum_{j=1}^{n'} p_{j}( p_1 x_{\{j\} \cup \{1\}} +  \dots + p_{j} x_{\{j\} \cup \{j\}}) = \frac{1}{2}\left(\sum_{j=1}^{n'} x_j p_j^2 +  \sum_{j,j'=1}^{n'} x_{\{j\} \cup \{j'\}} p_j p_{j'} \right).
  \end{align*} 
  As $x$ and $p$ are non-negative vectors, ignoring the terms $x_{\{j\}\cup \{j'\}}$ with $j\in S$ and $j'\in S'$, this can be lower bounded by
  \begin{align*}
    \underbrace{\frac{1}{2} \left(\sum_{j\not\in S} x_j p_j^2 +  \sum_{j,j' \not \in S} x_{\{j\} \cup \{j'\}} p_j p_{j'} \right)}_{(I)}
    +
    \underbrace{\frac{1}{2}\left(\sum_{j\in S} x_j p_j^2 +  \sum_{j,j' \in S} x_{\{j\} \cup \{j'\}} p_j p_{j'} \right)}_{(II)}.
  \end{align*}
  Again using that $x$ and $p$ are non-negative, ignoring the terms with $j\neq j$ we also have that
  \begin{align*}
    \left( \sum_{j, j'\not \in S}   p_{j} p_{j'} x_{\{j\} \cup \{j'\}} \right)  \geq  \sum_{j\not \in S} x_{j} p_j^2.
  \end{align*}
  Hence,  $(I) \geq \sum_{j\not \in S} x_j p_{j}^2$.

  Let us now concentrate on $(II)$ and in particular we show that
  \[\sum_{j,j'\in S} x_{\{j\} \cup \{j'\}} p_j p_{j'} \geq \mu^2\]
	where  $\mu =  \sum_{j\in S}x_j p_j$.

 To show this we use the PSD constraint on $X^{(i)}$. Let $v$ be the
  $(|S|+1)$ dimensional vector indexed by $\emptyset$ and $\{ij\}_{j\in S}$ whose
  entries are defined by  $v_\emptyset = -\mu$ and $v_{ij} = p_{ij} = p_j$ for $j\in S$. Let also $\ov{X}^{(i)}$ be the principal submatrix of $X^{(i)}$  containing those rows and columns indexed by $\emptyset$ and $\{ij\}_{j\in S}$.
  Then 
  \begin{align*}
   v^T \ov{X}^{(i)} v & = X^{(i)}_{\emptyset,\emptyset} v^2_\emptyset + 2 \sum_{j \in S}  X^{(i)}_{j,\emptyset} v_j  v_{\emptyset}+  \sum_{j, j'\in S}  X^{(i)} X^{(i)}_{\{j\}, \{j'\}} v_j v_{j'} \\
	& =  x_{\emptyset} \mu^2 - \sum_j 2x_{\{j\}} p_j \mu + \sum_{j,j'\in S} x_{\{j\} \cup \{j'\}} p_j p_{j'} \\
	& = \mu^2 - 2\mu^2 +\sum_{j,j'\in S} x_{\{j\} \cup \{j'\}} p_j p_{j'}  =  \sum_{j, j'\in S} p_{ij} p_{ij'} x_{\{ij\} \cup \{ij'\}}  - \mu^2,
  \end{align*}
  which is greater than $0$ because of the constraint $X^{(i)} \succeq 0$ in
  our relaxation (and hence the submatrix $\ov{X}^{(i)}$ is also positive
  semidefinite). This shows that 
  $(II) \geq \frac{1}{2}\left(\sum_{j\in S} x_jp_j^2 + \left( \sum_{j\in S} x_j p_j \right)^2  \right)$ 
  and completes the proof of the lemma.
\end{proof}

Let $\mathrm{LB}(S)$  denote $\sum_{j\not \in S} x_j p_j^2 + \frac{1}{2}\left( \sum_{j\in S}x_j p_j^2 + \left( \sum_{j\in S} x_j p_j\right)^2 \right)$. By setting $S=\emptyset$, $G$ and $J$, the lemma directly implies the following lower bounds: 

\begin{corollary}
\label{cor:bounds}
  We have the following lower bounds on
  $\sum_{j=1}^{n'} p_{j}( p_1 x_{\{j\} \cup \{1\}} +  \dots + p_{j} x_{\{j\} \cup \{j\}}):$
  \begin{align*}
    \lbE  & = Q, \\
    \lbJ  & = \nicefrac{1}{2}(Q + L^2),\\
    \lbG  & = \nicefrac{1}{2}( \ov{Q} + Q + (L - \ov{L})^2).
  \end{align*}
  \label{cor:lowerbounds}
\end{corollary}

\subsubsection{Proof of Inequality~\eqref{eq:toprove2}: bounding the approximation guarantee}
\label{sec:finalproof}

We use Lemma~\ref{lem:upperbound} and Corollary~\ref{cor:lowerbounds} to
prove~\eqref{eq:toprove2}. Let $\epsilon = 1/100$. We divide the proof into two
cases. Intuitively, the first case is when we have ``few'' ungrouped jobs and
then we get an improvement from the $\zeta$ in Theorem~\ref{thm:genassignment}.
In the other case, when we have ``many'' ungrouped jobs, note that jobs of
different job classes have (informally) very different processing times and
thus does not affect each other. This together with that the total fractional
mass of ungrouped jobs in each class is less than $1/10$ actually gives that
a simple randomized rounding does better than the factor $3/2$.  The formal
proof of the two cases are as follows:
\begin{description}
  \item[Case  $\ov{L} \leq (1-\sqrt{\epsilon}) L:$] We will upper bound the LHS of ~\eqref{eq:toprove2})
    \begin{align}
		\label{algcost:case1}
      \left(1-  \epsilon\frac{\zeta}{100}\right) \lbJ + \left( \frac{1}{2} - \frac{\zeta}{100} + \epsilon \frac{\zeta}{200} \right) \lbE + \frac{\zeta}{100} \lbG{}.
    \end{align}
		By Corollary \ref{cor:bounds}  this is at most 
		\[ \left(1-  \epsilon\frac{\zeta}{100}\right) + \left( \frac{1}{2} - \frac{\zeta}{100} +  \epsilon \frac{\zeta}{200} \right) + \frac{\zeta}{100} =  \left( \frac{3}{2}
    -  \epsilon \frac{\zeta}{200} \right) \]
	times the RHS  of~\eqref{eq:toprove2}.

 By the definition of $\lbJ, \lbE$ and $\lbG$, \eqref{algcost:case1} can be written as  
    \begin{align*}
      & \left(1-  \epsilon\frac{\zeta}{100}\right) \left( \frac{Q}{2} + \frac{L^2}{2} \right) + \left( \frac{1}{2} - \frac{\zeta}{100} + \epsilon \frac{\zeta}{200} \right) Q + \frac{\zeta}{100} \left( \frac{\ov{Q} + Q + (L-\ov{L})^2}{2} \right) \\
      = & \left( 1 - \frac{\zeta}{200} \right)  Q + \frac{\zeta}{200} \ov{Q}  + \left( 1- \epsilon \frac{\zeta}{100} \right) \frac{L^2}{2} + \frac{\zeta}{100} \frac{(L - \ov{L})^2}{2} \\
      \geq & \left( 1 - \frac{\zeta}{200} \right)  Q + \frac{\zeta}{200} \ov{Q}  + \left( 1- \epsilon \frac{\zeta}{100} \right) \frac{L^2}{2} + \epsilon \frac{\zeta}{100} \frac{L^2}{2} \qquad \mbox{(since $\ov{L} \leq (1-\sqrt{\epsilon})L$)} \\
      = & \left( 1 - \frac{\zeta}{200} \right)  Q + \frac{\zeta}{200} \ov{Q}  + \frac{L^2}{2}, 
    \end{align*}
    which is the upper bound on the LHS of~\eqref{eq:toprove2} from
    Lemma~\ref{lem:upperbound} and thus completes this case.
		
  \item[Case $\ov{L} > (1-\sqrt{\epsilon}) L:$] Let $\ov{\mu}
    = \ov{L}/(\sum_{j\in \ov{G}} x_j)$ denote the expected job size in
    $\ov{G}$, i.e., of the ungrouped jobs,  and let $k$ denote the class of
    $\ov{\mu}$. Let $N \subseteq \ov{G}$ denote jobs in $\ov{G}$ in classes
    $k-1$ and higher. Also let $X(N) =  \sum_{j\in N} x_j$. 
    
    We claim that $x(N) \leq 1/2$. Indeed, by Markov's inequality, the total
    mass of jobs in $\ov{G}$ in classes $k+2$ or higher is at most $1/10$.
    Moreover, as the mass of each class in $\ov{G}$ is at most $1/10$, we get
    $x(N) \leq 3/10 + 1/10 \leq 1/2$.

    Let us define $L(N) = \sum_{j\in N} x_j p_j$ and $Q(N) = \sum_{j\in N} x_j
    p_j^2$. By Cauchy-Schwarz, we have 
    $(\sum_{j\in N} x_j p_j^2) (\sum_{j\in N} x_j)  \geq (\sum_{j\in N} x_j p_j)^2$ 
    and hence
    \begin{align*}
      Q(N) \geq \frac{L(N)^2}{x(N)} \geq 2 L(N)^2.
    \end{align*}
    Next we show that the total expected size of jobs in $\ov{G} \setminus N$
    is negligible compared to $L(N)$. Indeed a job of class $k-h$ has
    processing time at most $\ov{\mu}/10^{h-1}$ and the total mass of jobs of
    class $k-h$ in $\ov{G}$ is at most $\sum_{j\in \ov{G}} x_j$  since this is
    the total mass of all jobs in $\ov{G}$. This gives us the rough upper bound
    \begin{align*}
      \ov{L} - L(N) \leq \sum_{h=2}^\infty \frac{1}{10^{h-1}} \cdot \ov{\mu} \left( \sum_{j\in \ov{G}} x_j \right) = \frac{\ov{L}}{9}. 
    \end{align*}

    The above gives us that
    \begin{align*}
      L^2 &\leq \frac{\ov{L}^2}{(1-\sqrt{\epsilon})^2} 
      \leq \left(\frac{9}{8(1-\sqrt{\epsilon})}\right)^2 L(N)^2
      \leq \left(\frac{9}{8(1-\sqrt{\epsilon})}\right)^2 \frac{Q(N)}{2} \\
      &= \left( \frac{10}{8} \right)^2 \frac{Q(N)}{2} = \frac{25}{32}Q(N)\leq \frac{25}{32}Q.
    \end{align*}
    By Lemma~\ref{lem:upperbound},  we have that the LHS of~\eqref{eq:toprove2} is at most~\footnote{This bound is loose and can also be obtained using a independent randomized rounding instead of our negative correlation rounding. The importance of the new rounding appears in the other case.}
    \begin{align*}
      Q + \frac{L^2}{2} 
      \leq \left( 1+ \frac{25}{64} \right) Q  
      = \left( 1+ \frac{25}{64} \right) \lbE < \left( \frac{3}{2} - c \right)\lbE,
    \end{align*}
    which completes this case and the proof of~\eqref{eq:toprove2} (and thus
    Theorem~\ref{thm:main2}).
\end{description}
}


\section*{Acknowledgement}
This work was done in part while the first author was visiting the Simons Institute for the Theory of Computing.
We thank Nick Harvey and Bruce Shepherd for organizing the Bellairs Workshop on Combinatorial Optimization 2015, 
which was the starting point for this work. Aravind Srinivasan thanks Amit Chavan for his help with \LaTeX packages.


\bibliographystyle{plain}
\bibliography{lit}

\begin{thebibliography}{10}

\bibitem{Afrati99}
Foto~N. Afrati, Evripidis Bampis, Chandra Chekuri, David~R. Karger, Claire
  Kenyon, Sanjeev Khanna, Ioannis Milis, Maurice Queyranne, Martin Skutella,
  Clifford Stein, and Maxim Sviridenko.
\newblock Approximation schemes for minimizing average weighted completion time
  with release dates.
\newblock In {\em Foundations of Computer Science, {FOCS}}, pages 32--44, 1999.

\bibitem{AgeevS99}
Alexander~A. Ageev and Maxim Sviridenko.
\newblock Approximation algorithms for maximum coverage and max cut with given
  sizes of parts.
\newblock In {\em Integer Programming and Combinatorial Optimization {IPCO}},
  pages 17--30, 1999.

\bibitem{AroraFK02}
Sanjeev Arora, Alan~M. Frieze, and Haim Kaplan.
\newblock A new rounding procedure for the assignment problem with applications
  to dense graph arrangement problems.
\newblock {\em Math. Program.}, 92(1):1--36, 2002.

\bibitem{AsadpourFS12}
Arash Asadpour, Uriel Feige, and Amin Saberi.
\newblock Santa claus meets hypergraph matchings.
\newblock {\em {ACM} Transactions on Algorithms}, 8(3):24, 2012.

\bibitem{AsadpourS07}
Arash Asadpour and Amin Saberi.
\newblock An approximation algorithm for max-min fair allocation of indivisible
  goods.
\newblock In {\em Symposium on Theory of Computing, {STOC}}, pages 114--121,
  2007.

\bibitem{AzarE05}
Yossi Azar and Amir Epstein.
\newblock Convex programming for scheduling unrelated parallel machines.
\newblock In {\em Symposium on Theory of Computing}, pages 331--337, 2005.

\bibitem{BansalS06}
Nikhil Bansal and Maxim Sviridenko.
\newblock The santa claus problem.
\newblock In {\em Symposium on Theory of Computing, {STOC}}, pages 31--40,
  2006.

\bibitem{ChakrabartyCK09}
Deeparnab Chakrabarty, Julia Chuzhoy, and Sanjeev Khanna.
\newblock On allocating goods to maximize fairness.
\newblock In {\em Foundations of Computer Science, {FOCS}}, pages 107--116,
  2009.

\bibitem{ChekuriK01}
Chandra Chekuri and Sanjeev Khanna.
\newblock A {PTAS} for minimizing weighted completion time on uniformly related
  machines.
\newblock In {\em {ICALP}}, pages 848--861, 2001.

\bibitem{ChekuriK04}
Chandra Chekuri and Sanjeev Khanna.
\newblock Approximation algorithms for minimizing averageweighted completion
  time.
\newblock In {\em Handbook of Scheduling - Algorithms, Models, and Performance
  Analysis.} 2004.

\bibitem{ChekuriMNS01}
Chandra Chekuri, Rajeev Motwani, B.~Natarajan, and Clifford Stein.
\newblock Approximation techniques for average completion time scheduling.
\newblock {\em {SIAM} J. Comput.}, 31(1):146--166, 2001.

\bibitem{ChekuriVZ11}
Chandra Chekuri, Jan Vondr\'{a}k, and Rico Zenklusen.
\newblock Multi-budgeted matchings and matroid intersection via dependent
  rounding.
\newblock In {\em Symposium on Discrete Algorithms, {SODA}}, pages 1080--1097,
  2011.

\bibitem{Feige08}
Uriel Feige.
\newblock On allocations that maximize fairness.
\newblock In {\em Symposium on Discrete Algorithms, {SODA}}, pages 287--293,
  2008.

\bibitem{GandhiKPS06}
Rajiv Gandhi, Samir Khuller, Srinivasan Parthasarathy, and Aravind Srinivasan.
\newblock Dependent rounding and its applications to approximation algorithms.
\newblock {\em J. {ACM}}, 53(3):324--360, 2006.

\bibitem{Goemans02}
Michel~X. Goemans, Maurice Queyranne, Andreas~S. Schulz, Martin Skutella, and
  Yaoguang Wang.
\newblock Single machine scheduling with release dates.
\newblock {\em {SIAM} J. Discrete Math.}, 15(2):165--192, 2002.

\bibitem{HallSSW97}
Leslie~A. Hall, Andreas~S. Schulz, David~B. Shmoys, and Joel Wein.
\newblock Scheduling to minimize average completion time: Off-line and on-line
  approximation algorithms.
\newblock {\em Mathematics of Operations Research}, 22(3):513--544, 1997.

\bibitem{HallSW96}
Leslie~A. Hall, David~B. Shmoys, and Joel Wein.
\newblock Scheduling to minimize average completion time: Off-line and on-line
  algorithms.
\newblock In {\em Symposium on Discrete Algorithms, {SODA}}, pages 142--151,
  1996.

\bibitem{HoogeveenSW98}
Han Hoogeveen, Petra Schuurman, and Gerhard~J. Woeginger.
\newblock Non-approximability results for scheduling problems with minsum
  criteria.
\newblock In {\em Integer Programming and Combinatorial Optimization, {IPCO}},
  pages 353--366, 1998.

\bibitem{KahnK97}
Jeff Kahn and P.~Mark Kayll.
\newblock On the stochastic independence properties of hard-core distributions.
\newblock {\em Combinatorica}, 17(3):369--391, 1997.

\bibitem{KumarMPS08}
V.~S.~Anil Kumar, Madhav~V. Marathe, Srinivasan Parthasarathy, and Aravind
  Srinivasan.
\newblock Minimum weighted completion time.
\newblock In {\em Encyclopedia of Algorithms}. 2008.

\bibitem{KumarMPS09}
V.~S.~Anil Kumar, Madhav~V. Marathe, Srinivasan Parthasarathy, and Aravind
  Srinivasan.
\newblock A unified approach to scheduling on unrelated parallel machines.
\newblock {\em J. ACM}, 56(5), 2009.

\bibitem{LenstraST90}
Jan~Karel Lenstra, David~B. Shmoys, and {\'{E}}va Tardos.
\newblock Approximation algorithms for scheduling unrelated parallel machines.
\newblock {\em Math. Program.}, 46:259--271, 1990.

\bibitem{MakarychevS14}
Konstantin Makarychev and Maxim Sviridenko.
\newblock Solving optimization problems with diseconomies of scale via
  decoupling.
\newblock In {\em Foundations of Computer Science, {FOCS}}, pages 571--580,
  2014.

\bibitem{PhillipsSW97}
Cynthia~A. Phillips, Clifford Stein, and Joel Wein.
\newblock Task scheduling in networks.
\newblock {\em {SIAM} J. Discrete Math.}, 10(4):573--598, 1997.

\bibitem{PSW98}
Cynthia~A. Phillips, Clifford Stein, and Joel Wein.
\newblock Minimizing average completion time in the presence of release dates.
\newblock {\em Math. Program.}, 82:199--223, 1998.

\bibitem{SchulzS02}
Andreas~S. Schulz and Martin Skutella.
\newblock Scheduling unrelated machines by randomized rounding.
\newblock {\em {SIAM} J. Discrete Math.}, 15(4):450--469, 2002.

\bibitem{Schuurman99}
Petra Schuurman and Gerhard Woeginger.
\newblock Polynomial time approximation algorithms for machine scheduling: Ten
  open problems.
\newblock {\em Journal of Scheduling}, 2(5):203--213, 1999.

\bibitem{SethuramanS99}
Jay Sethuraman and Mark~S. Squillante.
\newblock Optimal scheduling of multiclass parallel machines.
\newblock In {\em {ACM-SIAM} Symposium on Discrete Algorithms, {SODA}}, pages
  963--964, 1999.

\bibitem{Skutella:personal}
Martin Skutella.
\newblock Personal communication. Oct 2015.

\bibitem{Skutella01}
Martin Skutella.
\newblock Convex quadratic and semidefinite programming relaxations in
  scheduling.
\newblock {\em J. {ACM}}, 48(2):206--242, 2001.

\bibitem{SkutellaW99}
Martin Skutella and Gerhard~J. Woeginger.
\newblock A {PTAS} for minimizing the weighted sum of job completion times on
  parallel machines.
\newblock In {\em Symposium on Theory of Computing, {STOC}}, pages 400--407,
  1999.

\bibitem{Smith56}
W.~E. Smith.
\newblock Various optimizers for single-stage production.
\newblock {\em Naval Research Logistics}, 3:59–66, 1956.

\bibitem{Svensson12}
Ola Svensson.
\newblock Santa claus schedules jobs on unrelated machines.
\newblock {\em {SIAM} J. Comput.}, 41(5):1318--1341, 2012.

\bibitem{SviridenkoW13}
Maxim Sviridenko and Andreas Wiese.
\newblock Approximating the configuration-lp for minimizing weighted sum of
  completion times on unrelated machines.
\newblock In {\em Integer Programming and Combinatorial Optimization, {IPCO}},
  pages 387--398, 2013.

\end{thebibliography}

\end{document}